\documentclass[12pt,technote, onecolumn, draft]{IEEEtran}
\usepackage[utf8]{inputenc}
\usepackage{latexsym}
\usepackage{mathrsfs}
\usepackage{graphicx}
\usepackage{multirow}
\usepackage{amsfonts,amssymb,amsmath,amsthm,bm}
\usepackage{color}
\usepackage{booktabs}
\usepackage{cite}
\newtheorem{theorem}{Theorem}[section] 
\newtheorem{definition}[theorem]{Definition} 
\newtheorem{lemma}[theorem]{Lemma} 
\newtheorem{corollary}[theorem]{Corollary}
\newtheorem{example}[theorem]{Example}
\newtheorem{proposition}[theorem]{Proposition}
\newtheorem{remark}[theorem]{Remark}

\begin{document}
	\title{Hyper-derivative Algebraic Geometry Codes via Local Expansions}
	\author{Xiaofeng Liu, Hengfeng Liu, Jun Zhang, Fang-Wei Fu, Chunming Tang
		\IEEEcompsocitemizethanks{\IEEEcompsocthanksitem Xiaofeng Liu and Fang-Wei Fu are with the Chern Institute of Mathematics and LPMC, Nankai University, Tianjin 300071, China. Jun Zhang is with the School of Mathematical Sciences, Capital Normal University, Beijing 100048, China. Hengfeng Liu is with the School of Mathematics, Southwest Jiaotong University, Chengdu 610031, China, and Chunming Tang is with the School of Information Science and Technology, Southwest Jiaotong University, Chengdu 610031, China.      
			Emails: lxfhah@mail.nankai.edu.cn, hengfengliu@163.com, junzhang@cnu.edu.cn, fwfu@nankai.edu.cn, tangchunmingmath@163.com.
		}
		\thanks{Xiaofeng Liu and Fang-Wei Fu were supported by the National Key Research and Development Program of China (Grant Nos.  2022YFA1005000), the National Natural Science Foundation of China (Grant Nos. 12141108, 61971243), the Fundamental Research Funds for the Central Universities of China (Nankai University), and the Nankai Zhide Foundation.
			Jun Zhang was supported by the National Natural Science Foundation of China under Grant Nos. 12222113, 12441105.        
			The research of Hengfeng Liu and Chunming Tang was supported by the National Natural Science Foundation of China under grant No. 12231015, and by the Science and Technology Projects of Xizang Autonomous Region, China, under Grant XZ202502JD0036.}
	}{\tiny }
	\maketitle	
	
	\begin{abstract}
		Hyper-derivative Reed–Solomon (HRS) codes are a generalization of Reed–Solomon codes under the Niederreiter–Rosenbloom–Tsfasman (NRT) metric. In this paper, we develop a systematic construction framework of hyper-derivative algebraic geometry (HAG) codes via local expansions, extending HRS codes from the rational function field to general algebraic function fields. Using the residue theorem, we determine their Euclidean duals and illustrate that the duals naturally reverse. We further give criteria for reverse self-orthogonality and reverse self-duality in terms of two classes of bilinear forms. Finally, we provide an asymptotic bound on the rate and relative NRT distance via function field towers and Ihara’s constant.
	\end{abstract}

	
	\section{Introduction}
	\label{sec:1}
	Niederreiter-Rosenbloom-Tsfasman (NRT) metric was initially introduced in 1987 to study uniform distribution of point sets in Euclidean spaces.  NRT Reed-Solomon (NRT RS) codes introduced by Rosenbloom and Tsfasman serve as a class of original maximally distance separable (MDS) codes.
	Recently, Can and Horowitz introduced Hyper-derivative Reed-Solomon (HRS) codes; see \cite{11}. HRS codes form an important class of MDS codes under the NRT metric, which can be seen as a generalization of the NRT RS codes. In that paper, Can and Horowitz also studied the duality and low density property, \textit{etc.} for HRS codes. They stated that the Euclidean dual of an HRS code is still an HRS code. However, such a property does not hold for general HRS codes. More recently, the actual duality structure and reverse self-dual constructions of HRS codes were established in \cite{99}. 

   Function fields with many rational places allow us to design codes with lengths beyond RS codes.
    Over the past decades, function fields have also found applications in several scenarios of coding theory.  Previously, NRT RS codes were generalized to function fields by Tsfasman \textit{et. al.}; see chapter 12 in \cite{8}. Such codes are called NRT algebraic geometry (NRT AG) codes. It is therefore natural to explore whether the duality structure and some fundamental properties of HRS codes can be extended to general function fields.  Recently, Can \textit{et. al.} extended the evaluation codes under bottleneck metrics to AG codes in \cite{50} and studied some properties for such codes. This leads us to study the properties of hyper-derivative algebraic geometry (HAG) codes systematically.

    Local expansion serves as a main technical tool in our construction of HAG codes and analysis of their duality. It has been widely applied in coding theory and cryptography, including the constructions of digital nets \cite{5,6}, algebraic geometry codes \cite{7}, and sequences with almost perfect linear complexity profiles \cite{10}, as well as distributed matrix multiplication \cite{4}.

	This paper is organized as follows. In Section \ref{sec:2}, we will provide some fundamental results of function fields over finite fields, local expansions, Ihara's constant and Garcia-Stichtenoth towers. In Section \ref{sec:6}, we introduce the main results and detailed comparison. In Section \ref{sec:3}, we introduce two constructions of HAG codes via local expansions of functions and differentials. Then we determine the duality and reverse dual codes. By carefully choosing a special type of differentials, we also provide a normalized construction of HAG codes. In Section \ref{sec:4}, we obtained some asymptotic results by using some function field towers. Finally, we make a conclusion and list some possible future work in Section \ref{sec:5}.
	
	\section{Preliminaries }
	\label{sec:2}
	In this section, we mainly recall some fundamental results of global function fields, local expansions, Garcia-Stichtenoth tower, Ihara's constant and NRT metric.
	\subsection{Function Fields over Finite Fields}
	Let $F/\mathbb{F}_{q}$ be a function field over finite field $\mathbb{F}_{q}$ with genus $\mathfrak{g}_{F}$. Denote by $\mathbb{P}_{F}$ and $\mathbb{D}_{F}$ the set of all places and divisors of $F$.    Any discrete valuation $\mathrm{v}_{P}$ corresponds to some place $P\in\mathbb{P}_{F}$. For any nonzero divisor $G=\sum_{P\in\mathbb{P}_{F}}\mathrm{v}_{P}(G)P$, the support of $G$ is defined by $\mathrm{Supp}(G)=\{P\in\mathbb{P}_{F}| \mathrm{v}_{P}(G)\neq 0\}$. For any nonzero function $f\in F$, we denote by $(f),(f)_{0},(f)_{\infty}$ the principal divisor, zero divisor and pole divisor of $f$, respectively.  We also have
	$(f)=(f)_{0}-(f)_{\infty}$ with the decomposition $(f)_{0}=\sum_{P\in\mathbb{P}_{F},\mathrm{v}_{P}(f)>0}\mathrm{v}_{P}(f)P$ and $
	(f)_{\infty}=\sum_{P\in\mathbb{P}_{F},\mathrm{v}_{P}(f)<0}(-\mathrm{v}_{P}(f))P$. The set of principal divisors forms a group under addition and is denoted by $\mathrm{Princ}(F)=\{(f): f\in F\}$.
    We call two divisors $A$ and $B$ in $\mathbb{D}_{F}$ equivalent if there exists $z\in F$ such that $A = B + (z)$ and we denote by $A\sim B$.  The factor group $\mathrm{Pic}(F) = \mathbb{D}_{F}/\mathrm{Princ}(F)$ is called the Picard group of $F/\mathbb{F}_{q}$.
	
	For any non-negative divisor $G\in\mathbb{D}_{F}$, the Riemann-Roch space $\mathcal{L}(G)$ is defined as
	\begin{displaymath}
		\mathcal{L}(G):=\{f\in F\setminus\{0\}: (f)\geq -G\}\cup\{0\},
	\end{displaymath}
	which is a finite $\mathbb{F}_{q}-$dimensional vector space. We denote by $\ell(G)$ the dimension of $\mathcal{L}(G)$. 
	
	Denote by $\Omega_{F}$ the space consisting of all differentials of $F$. Then we have the differential space $\Omega_{F}(G)=\{\omega\in\Omega_{F}: (\omega)\geq G\}\cup\{0\}$ which is also a finite $\mathbb{F}_{q}-$dimensional vector space. We denote it by $i(G)=\dim_{\mathbb{F}_{q}}\Omega_{F}(G)$. By the Riemann-Roch theorem, we have $\ell(G)-i(G) =\deg(G)+1-\mathfrak{g}_{F}$.
	
	\subsection{Local Expansions}
	Let $\lambda$ be a local uniformizer of a place $P\in\mathbb{P}_{F}$, \textit{i. e.}, $\mathrm{v}_{P}(\lambda)=1$. For any function $f\in F^*$ with $\mathrm{v}_{P}(f)\geq t$, we have $\mathrm{v}_{P}(f/\lambda^{t})\geq 0$. Put $a_{t}=f/\lambda^{t}(P)$.  Then $f/\lambda^{t}-a_{t}$ satisfies $\mathrm{v}_{P}(f/\lambda^{t}-a_{t})\geq 1$ and $\mathrm{v}_{P}((f-a_{t}\lambda^{t})/\lambda^{t+1})\geq 0$. Denote by $a_{t+1}=((f-a_{t}\lambda^{t})/\lambda^{t+1})(P)$. By such an induction process, we shall obtain a sequence $\{a_{r}\}_{r=t}^{\infty}$ that is equivalent to the formal expansion of $f$, \textit{that is}, $f=\sum^{\infty}_{i=t}a_{i}\lambda^{i}$, and we call it a local expansion of $f$ at $P$.
	
	\subsection{Ihara's Constant $A(q)$}
	Let $\mathfrak{g}$ be a positive integer and we denote by $N_{q}(\mathfrak{g})$ the maximum number of rational places of any function field with constant field $\mathbb{F}_{q}$ and genus $\mathfrak{g}$. Define the real number 
	\begin{displaymath}
		A(q)=\lim\sup_{\mathfrak{g}\to\infty}\frac{N_{q}(\mathfrak{g})}{\mathfrak{g}}
	\end{displaymath}
	and it is called Ihara's constant. 
 The Drinfeld-Vl{\u{a}}du{\c{t}} in \cite{8}  holds
	\begin{displaymath}
		A(q)\leq\sqrt{q}-1.
	\end{displaymath}
	In particular, if $q$ is a square power of a prime, the equality holds via modular curves; see \cite{17},
	\begin{displaymath}
		A(q)=\sqrt{q}-1.
	\end{displaymath}

	If $q=p^{2m+1}$ is an odd power of a prime, Bassa \textit{et. al.} in \cite{18} obtained the following lower bound and also provided an explicit construction to attain such bound,
	\begin{displaymath}
		A(p^{2m+1})\geq\frac{2(p^{m+1}-1)}{p+1+\zeta},\ \zeta=\frac{p-1}{p^{m}-1}.
	\end{displaymath}
	If $\ell$ is a prime, $\ell\mid (p-1)$, and $p>4\ell+1$, then the following bound holds; see \cite{8}.
	\begin{displaymath}
		A(p^{\ell})\geq\frac{\sqrt{\ell(p-1)}-2\ell}{\ell-1}.
	\end{displaymath}

	\subsection{Garcia-Stichtenoth Tower}
	Let $q$ be a square of a prime power. The Garcia-Stichtenoth tower $\mathcal{F}=(F_{i}/\mathbb{F}_{q})_{i\geq 1}$ with $F_{1}=\mathbb{F}_{q}(y_{1})$ and the iterative process $F_{i}=F_{i-1}(y_{i})$ is defined by
	\begin{displaymath}
		y_{i}^{\sqrt{q}}+y_{i}=\frac{y^{ \sqrt{q} }_{i-1}}{y^{  \sqrt{q}  -1}_{i-1}+1}
	\end{displaymath}
	for $i\geq 2$. The number of rational points satisfies $N(F_{i})\geq (q- \sqrt{q}   ) \sqrt{q}  ^{i-1}+ \sqrt{q}  $ and the genus of $F_{i}$ is divided into the two cases
			\begin{displaymath}
				\mathfrak{g}_{F_{i}}=\begin{cases}
					(  \sqrt{q} ^{\frac{i}{2}}-1)^{2},&\text{if}\ i\equiv 0\ (\text{mod}\ 2),\\
					(  \sqrt{q}  ^{\frac{i+1}{2}}-1)( \sqrt{q}   ^{\frac{i-1}{2}}-1),&\text{if}\ i\equiv 1\ (\text{mod}\ 2).
				\end{cases}
			\end{displaymath}
	The asymptotic parameter can be given by $\lim_{i\to\infty}\frac{N(F_{i})}{\mathfrak{g}_{T_{i}}}=\sqrt{q}-1$ which means $\mathcal{F}$ is asymptotically optimal. For details, the readers may refer to \cite{12}.

	\subsection{NRT Metric}
	Different from the classical error-correcting codes under the Hamming metric. The NRT metric is defined on $\mathbb{F}_{q}^{s\times r}$. It is a poset (partially ordered set) metric. First, we need the NRT metric for a column vector. Denote by $[t]=\{1,2,\cdots,t\}$.

	For any codeword $\mathbf{c}=(\mathbf{c}_{1},\mathbf{c}_{2},\cdots,\mathbf{c}_{r})\in\mathbb{F}_{q}^{s\times r}$, its NRT weight is defined by the sum of NRT weights of all columns, \textit{i. e.},
	\begin{displaymath}
		w_{NRT}(\mathbf{c})=\sum^{r}_{j=1}w_{NRT}(\mathbf{c}_{j}).
	\end{displaymath}
 where the NRT weight of a column $\mathbf{c}_{j}= (c_{j,1},c_{j,2},\cdots,c_{j,s})\in\mathbb{F}_{q}^{s} , 1\leq j\leq r$ is defined by
	\begin{displaymath}
		w_{NRT}(\mathbf{c}_{j})=\begin{cases}
			s-i+1& \text{if}\ \mathbf{c}_{j}\neq \mathbf{0}\\
			0&\text{if}\ \mathbf{c}_{j}=\mathbf{0}.
		\end{cases}
	\end{displaymath}
    where $  i=\min\{b\in [s]: c_{j,b}\neq 0\}  $.
    
	The NRT distance between two codewords $A\neq B\in\mathbb{F}_{q}^{s\times r}$ is defined by:
	\begin{displaymath}
		d_{NRT}(A,B)=w_{NRT}(A-B).
	\end{displaymath}
If $\mathcal{C}$ is a linear code with parameters $[rs, k]_{q}$ under NRT metric, the minimum distance satisfies $d_{NRT}\leq rs-k + 1$; see \cite{51}.
     
\begin{definition}
    A linear code $C \subseteq\mathbb{F}_{q}^{s\times r}$   with parameters $[rs, k, d_{NRT}]_{q}$ under NRT metric is called an MDS NRT-metric code if it satisfies $d_{NRT}(\mathcal{C}) = rs- k + 1$.
\end{definition}

\section{Main Results and Techniques}
\label{sec:6}
    In this section, we introduce the main results and techniques used throughout this paper.
    
    First, the main tools are taken from local expansions and residue theorem of function fields. The duality theorem of HAG codes is deduced from the residue theorem; see Theorem \ref{t1}. We illustrate the reverse duality of HAG codes via local expansions for a given differential; see Theorem \ref{c1}. By carefully choosing a normalized differential, the conditions of reverse duality can be simplified; see Theorem \ref{cc}. In particular, we also determine the self-orthogonal and self-dual conditions in Theorem \ref{c3}. Finally, we provide some asymptotic results for HAG codes via function field towers; see Theorem \ref{c4}.  Table \ref{t2} lists the main objects of both HRS codes and HAG codes for comparison with prior constructions of HRS codes.
   \begin{table}[htbp]  
  \centering  
   \caption{Parameters of HRS and HAG codes}  
   \begin{tabular}{|c|c|c|}  
    \toprule
    Objects &HRS codes in \cite{99}& HAG codes  \\    \midrule
    Evaluation points& elements $\alpha_{j}, 1\leq j\leq r$& rational places  \\\midrule
   Evaluation function & $f\in\mathbb{F}_{q}[x]_{<k}$& $f\in\mathcal{L}(G)$  \\    \midrule
   Evaluation divisor&$\prod^{r}_{j=1}(x-\alpha_{j})^{s}$&$E=s\sum^{r}_{j=1}P_{j}$\\\midrule
    Dimension & $k$& $\ell(G)-\ell(G-E)$ \\    \midrule
    Minimum distance  & $d_{NRT}=sr-k+1$& $d_{NRT}\geq sr-\deg(G)$ \\ \midrule
    Generalized NRT distance&None&$d^{NRT}_{m}\geq\deg(E)-\deg(G)+\gamma_{m}(\mathcal{X})$\\\midrule
    Differential form&$ f(x)g(x)A(x)^{-1}dx  $&$f\omega$ or $fg\eta$\\\midrule
    Local multiplication & upper-triangular matrices $U_{j}, 1\leq j\leq r$&  upper-triangular matrices via local expansions\\\midrule
    Duality& reverse derivative of $g(x)/A_{j}(x)$& differentials in $\Omega_{F}(G-E)/\Omega_{F}(G)$\\\midrule
     Asymptotic results&None& $R+\delta>1-1/(s\lambda(F))$\\\midrule
   \bottomrule
  \end{tabular}
  \label{t2}  
\end{table}

	\section{hyper-derivative algebraic geometry codes and their duals}
	\label{sec:3}
	\subsection{Construction via Local Expansions}
	Let $F/\mathbb{F}_{q}$ be a function field with genus $\mathfrak{g}_{F}$. Take pairwise different rational places $D=\sum^{r}_{j=1}P_{j}$ and let $E=sD$. Take another divisor $ G\in\mathbb{D}_{F}$ with $\mathrm{Supp}(G)\cap\mathrm{Supp}(D)=\emptyset$.  Then we  have a local expansion of $f\in\mathcal{L}(G)$ at each $P_{j}, 1\leq j\leq r$. They can be given by  $$f=\sum_{a\geq 0}f_{j,a}t^{a}_{j}, \ f_{j,a}\in\mathbb{F}_{q}$$ for $1\le j\leq r$. 
	\begin{remark}
		Denote by $\mathcal{O}_{P_{j}}$ and $\mathbf{m}_{P_{j}}$ the valuation ring and the maximal ideal corresponding to place $P_{j}$, respectively. Then we have the isomorphism 
		\begin{displaymath}
		\mathcal{O}_{P_{j}}/\mathbf{m}^{s}_{P_{j}}\simeq\mathbb{F}_{q}[t_{j}]/(t^{s}_{j})
		\end{displaymath}
		for $1\leq j\leq r$.
	\end{remark}
	
	\begin{remark}
	In particular, the coefficients $f_{j,a},1\leq j\leq r, a\geq 0$ are Hasse derivatives $\partial^{(a)}f(\alpha_{j})$ if $F$ is a rational function field and the local parameter is chosen by $t_{j}=x-\alpha_{j}$ for $\alpha_{j}\in\mathbb{F}_{q}$; see \cite{9}.  
\end{remark}

	In addition, we also need a multiplier matrix $V=(v_{i,j})_{s\times r}$ with each element $v_{i,j}\in\mathbb{F}_{q}^{*}.$  Then the first construction of HAG code, \textit{i. e.}, evaluation HAG code can be given in the following  
	\begin{definition}\label{d1}
		The evaluation HAG  code $\mathcal{C}^{(s)}_{\mathcal{L}}(D,G,V)$ is defined as the image of the following evaluation map
		\begin{displaymath}
			\begin{split}
				\mathrm{ev}^{(s)}_{D,G,V}:\mathcal{L}(G)&\to\mathbb{F}_{q}^{s\times r}\\
				f&\mapsto (v_{a+1,j}f_{j,a})_{0\leq a\leq s-1, 1\leq j\leq r}.
			\end{split}
		\end{displaymath}
        When $V$ is the all-ones matrix, the evaluation HAG
codes reduce to the construction in \cite{50}.
	\end{definition}
	
	\begin{proposition}
		The dimension of $\mathcal{C}^{(s)}_{\mathcal{L}}(D,G,V)$ is given by
		\begin{displaymath}
			\dim\mathcal{C}^{(s)}_{\mathcal{L}}(D,G,V)=\ell(G)-\ell(G-E).
		\end{displaymath}
		In particular, we have $\ell(G-E)=0$ if $\deg G< sr$, which induces $  	\dim\mathcal{C}^{(s)}_{\mathcal{L}}(D,G,V)=\ell(G)$.
		\begin{proof}
			We only need to prove $\ker\mathrm{ev}^{(s)}_{D,G,V}=\mathcal{L}(G-E).$ The local coefficients $f_{j,a}=0$ for any $0\leq a\leq s-1$ and $1\leq j\leq r$ means $\mathrm{v}_{P_{j}}(f)\geq s$ for any $1\leq j\leq r$, \textit{i. e.}, $(f)+G-E\geq 0$. Then we have $\ker\mathrm{ev}^{(s)}_{D,G,V}=\mathcal{L}(G-E).$ By the Riemann-Roch theorem,  it follows immediately that $\ker\mathrm{ev}^{(s)}_{D,G,V}=\{0\}$   if $\deg G< sr$. 
		\end{proof}
	\end{proposition}
	\begin{remark}
		For any codeword $\mathrm{ev}_{D,G,V}^{(s)}(f)$, the NRT weight is given by 
		\begin{displaymath}
			w_{NRT}( \mathrm{ev}_{D,G,V}^{(s)}(f) )=sr-\sum^{r}_{j=1}\min\{\mathrm{v}_{P_{j}}(f),s\}.
		\end{displaymath}
	\end{remark}
		Note that NRT weight of the codeword $\mathrm{ev}_{D,G,V}^{(s)}(f)$ is independent of the choices of local parameters.

	The following theorem determines the minimum NRT distance for the code  $  	\mathcal{C}^{(s)}_{\mathcal{L}}(D,G,V)$, and the proof easily follows from the Riemann–Roch theorem.
	\begin{theorem}  \label{tt}
		Suppose $\deg(G)<sr$, then we have
		\begin{displaymath}
			\dim\mathcal{C}^{(s)}_{\mathcal{L}}(D,G,V)=\ell(G),\ \  d_{NRT}\left( \mathcal{C}^{(s)}_{\mathcal{L}}(D,G,V) \right)\geq sr-\deg(G).
		\end{displaymath}
		Furthermore,
		\begin{displaymath}
			k=\deg G+1-\mathfrak{g}_{F},\ \ d_{NRT}\geq sr-k+1-\mathfrak{g}_{F}
		\end{displaymath}
        if $\deg G\geq 2\mathfrak{g}_{F}-1$.
	\end{theorem}

    In the next subsection, we utilize gonality sequences to improve the result of minimum NRT distance above.
	\subsection{The Generalized Minimum NRT Distance}
	Motivated by the study of the generalized minimum Hamming distance for the AG codes and generalized weights for poset metrics in \cite{14} \textit{etc.}, we also investigate the generalized minimum NRT distance for HAG codes.

    For a given function field $F/\mathbb{F}_{q}$, we denote by $\mathcal{X}/\mathbb{F}_{q}$ the corresponding algebraic curve. The gonality sequence for a curve $\mathcal{X}/\mathbb{F}_{q}$ is defined as
    \begin{displaymath}
        \gamma_m(\mathcal{X}) = \min\{\deg A : \ell(A) \geq m\},\ m \geq 1.
    \end{displaymath}
	Note that $\gamma_{1}=0$ and $\gamma_{m}=m-1$ if we take $\mathcal{X}/\mathbb{F}_{q}$ as the projective line $\mathbb{P}^{1}_{\mathbb{F}_{q}}$. Assume $\deg G<sr$. For a subspace $W=\mathrm{ev}^{(s)}_{D,G,V}(U)$ with dimension $m=\dim_{\mathbb{F}_{q}}(U)$, we choose
\begin{displaymath}
     e_{j}(U)=\min\left\lbrace s, \min_{0\neq f\in U}\mathrm{v}_{P_{j}}(f)\right\rbrace,\ E(U)=\sum e_{j}(U)P_{j}
\end{displaymath}
then the generalized minimum NRT weight is defined as
\begin{displaymath}
   d_{m}^{NRT}(\mathcal{C})=\min\left\lbrace w^{NRT}_{m}(W),\ W  \subseteq\mathcal{C}, \dim_{\mathbb{F}_{q}}W = m\right\rbrace.
\end{displaymath}
  for $w^{NRT}_{m}(W)\geq sr-\deg E(U)$.  
    \begin{theorem}
        Assume $\deg G<sr$. Let $k=\ell(G)$. For $1\leq m \leq k$, the $m$th generalized minimum NRT distance is lower bounded by
        \begin{displaymath}
            d^{NRT}_{m} \left(\mathcal{C}_{\mathcal{L}}^{(s)}(D,G,V)\right) \geq sr-\deg G+\gamma_{m}(\mathcal{X}).
        \end{displaymath}
        In particular, we have
         \begin{displaymath}
            d^{NRT}_{m} \left(\mathcal{C}_{\mathcal{L}}^{(s)}(D,G,V)\right) = sr-k+m,
        \end{displaymath}
        if we take $\mathcal{X}/\mathbb{F}_{q}=\mathbb{P}^{1}_{\mathbb{F}_{q}}$ and $G=(k-1)P_{\infty}$.
\begin{proof}
    By the injectivity of the evaluation map $e^{(s)}_{D,G,V}$, it is sufficient to consider the subspace $U\subseteq\mathcal{L}(G)$. Let
    \begin{displaymath}
        e_{j}(U)=\min\left\lbrace s, \min_{0\neq f\in U}\mathrm{v}_{P_{j}}(f)\right\rbrace, E(U)=\sum e_{j}(U)P_{j}.
    \end{displaymath}
    Then we have $U\subseteq\mathcal{L}(G-E(U))$ and $\ell(G-E(U))\geq m$. Furthermore, we have $\deg(G-E(U))\geq\gamma_{m}(\mathcal{X})$, which is equivalent to $ \deg E(U)\leq \deg G-\gamma_m(\mathcal{X}) $. On the other hand, for each $j$, the support weight of $U$ at $P_{j}$ is given by $s-e_{j}(U)$. Then we have
    \begin{displaymath}
        w_{m}^{NRT}(U)\geq sr- \deg E(U) \geq sr- \deg G + \gamma_{m}(\mathcal{X}).
    \end{displaymath}
Consider the projective line $\mathbb{P}_{\mathbb{F}_{q}}^{1}$ and substitute $\gamma_{m} = m- 1$. Then we have $ d^{NRT}_{m} \left(\mathcal{C}_{\mathcal{L}}^{(s)}(D,G,V)\right)\geq sr-k+m   $. To attain it, the construction of an $m-$dimensional subspace $U$ satisfying $\deg E(U)=k-m$ can be achieved by choosing nonnegative integers $e_{j} \leq s$ with $\sum_{j}e_j = k - m$ and $h(x) = \prod_{j}(x - \alpha_{j} )^{e_{j}}$. We can take $U=h\mathcal{L}((m-1)P_{\infty})$. Then we have $\dim_{\mathbb{F}_{q}}U=m$ and $U\subseteq\mathcal{L}((k-1)P_{\infty})$.
 Then we have the upper bound.
\end{proof}
    \end{theorem}

	\subsection{Duality and Residue Theorem}

First, we need the following definition of Euclidean inner product to study duality.
	For any two matrices $A=(a_{i,j}), B=(b_{i,j})\in\mathbb{F}_{q}^{s\times r}$, the Euclidean inner product between $A$ and $B$ is defined as
	\begin{displaymath}
		\left<A,B\right>_{E}=\sum^{s}_{i=1}\sum^{r}_{j=1}a_{i,j}b_{i,j}.
	\end{displaymath}
	Then the Euclidean dual of a given linear code $\mathcal{C}\subseteq\mathbb{F}_{q}^{s\times r}$ is given by
	\begin{displaymath}
		\mathcal{C}^{\perp}=\lbrace B\in\mathbb{F}_{q}^{s\times r}: \left<A,B\right>_{E}=0\ \text{for any}\ A\in\mathcal{C}\rbrace.
	\end{displaymath}

	In this subsection, we will determine the dual code of $ \mathcal{C}^{(s)}_{\mathcal{L}}(D,G,V)$ under the Euclidean inner product defined above. Before we prove the duality theorem rigorously, we need to introduce the second construction of HAG codes.

	Take one differential form $\omega\in\Omega_{F}(G-E)$. The corresponding Laurent series of $\omega$ at the place $P_{j}$ can be given by
	\begin{displaymath}
		\omega=\left(\sum_{n\geq -s}c_{j,n}(\omega)t_{j}^{n}\right)dt_{j}, \ c_{j,n}(\omega)\in\mathbb{F}_{q}.
	\end{displaymath}
	Then we have the second construction of HAG code, \textit{i. e.}, differential HAG code $\mathcal{C}^{(s)}_{\Omega_{F}}(D,G,V)$.
	\begin{definition}\label{d2}
		The differential HAG code $ \mathcal{C}^{(s)}_{\Omega_{F}}(D,G,V)$ is defined as the image of following residue map
		\begin{displaymath}
			\begin{split}
				\mathrm{res}_{D,G,V}^{(s)}:\Omega_{F}(G-E)&\to\mathbb{F}_{q}^{s\times r}\\
				\omega&\mapsto (v_{a+1,j}c_{j,-a-1}(\omega))_{0\leq a\leq s-1, 1\leq j\leq r}
			\end{split}
		\end{displaymath}
		\textit{i. e.} $\mathcal{C}^{(s)}_{\Omega_{F}}(D,G,V)=\lbrace  (v_{a+1,j}c_{j,-a-1}(\omega))_{0\leq a\leq s-1, 1\leq j\leq r}: \omega\in\Omega_{F}(G-E)   \rbrace$.
	\end{definition}
	The following theorem determines duality between two constructions of HAG codes.
	\begin{theorem}\label{t1}
		For any non-negative divisor $G$ with $\mathrm{Supp}(G)\cap\mathrm{Supp}(D)=\emptyset$, we have
		\begin{displaymath}
			\left(\mathcal{C}^{(s)}_{\mathcal{L}}(D,G,V)\right)^{\perp}=\mathcal{C}^{(s)}_{\Omega_{F}}(D,G,V^\vee),
		\end{displaymath}
		for $  V^\vee=(v^{-1}_{a+1,j})_{0\leq a\leq s-1, 1\leq j\leq r}  $.
		\begin{proof}
			Take $0\neq f\in\mathcal{L}(G)$ and $0\neq \omega\in\Omega_{F}(G-E)$. For any $1\leq j\leq r$, we have
			\begin{displaymath}
				\sum^{s-1}_{a=0}f_{j,a}c_{j,-a-1}(\omega)=\mathrm{res}_{P_{j}}(f\omega).
			\end{displaymath}
			
			Consider divisors $(f)$ and $(\omega)$. Then we have
			\begin{displaymath}
				(f)+(\omega)\geq -G+(G-E)=-E,
			\end{displaymath}
			which means $f\omega$ is regular except at the support of $D$.
			
			By the residue theorem, we have the following duality
			\begin{displaymath}
				\langle\mathrm{ev}^{(s)}_{D,G,V}(f), \mathrm{res}_{D,G,V^\vee} ^{(s)}(\omega)\rangle_{E}=\sum^{r}_{j=1}\mathrm{res}_{P_{j}}(f\omega)=0,
			\end{displaymath}
			for any nonzero function $f\in\mathcal{L}(G)$.
			
			Note that the kernel of $  \mathrm{res}_{D,G,V} ^{(s)}$ is $\Omega_{F}(G)$. Then we have  $ \dim_{\mathbb{F}_{q}}\mathcal{C}^{(s)}_{\Omega_{F}}(D,G,V)=i(G-E)-i(G)$. By the Riemann-Roch theorem, we have
			\begin{displaymath}
				\begin{split}
					i(G-E)-i(G)&=\ell(G-E)-\deg(G-E)-1+\mathfrak{g}_{F}-\ell(G)+\deg(G)+1-\mathfrak{g}_{F}\\
					&=sr-(\ell(G)-\ell(G-E)).
				\end{split}
			\end{displaymath}
			Then we have the desired result.
		\end{proof}
	\end{theorem}
	
	\subsection{Reverse Duality for a Given Differential}
    In this subsection, the duality for a given differential is determined.
	Fix a nonzero differential $\eta\neq 0$ with $\mathrm{v}_{P_{j}}(\eta)=-s$. Define a dual divisor
	\begin{displaymath}
		G^{\vee}=(\eta)-G+E.
	\end{displaymath}

\begin{lemma}\label{lp}
    Any differential  $\omega\in\Omega_{F}(G-E)$ can be written as $\omega=h\eta$ uniquely for some $h\in\mathcal{L}(G^{\vee})$.
    \begin{proof}
    First, it is verified in Proposition 1.5.9 in \cite{1} that $\Omega_{F}$ is a one-dimensional vector space over $F$. Then for any $\omega\in\Omega_{F}(G-E)$, $\omega$ can be represented as $\omega=h\eta$ for a given differential $\eta$ and $h\in F^{*}$.
    It is equivalent to verify  $\omega\in\Omega_{F}(G-E)\iff h\in\mathcal{L}(G^{\vee})$.
        First, $(\omega)=(h\eta)=(h)+(\eta)\geq G-E$ which means $(h)\geq G-E-(\eta)=-G^{\vee}$. Then we have $h\in\mathcal{L}(G^{\vee})$. Uniqueness: If $h_{1}\eta=h_{2}\eta$, then we have $(h_{1}-h_{2})\eta=0$. Since $\eta\neq 0$, we have $h_{1}=h_{2}$ immediately. Then we have $\mathcal{L}(G^{\vee})\to\Omega_{F}(G-E)$ is a $\mathbb{F}_{q}$-linear isomorphism.
    \end{proof}
\end{lemma}

    Denote by the reverse of any vector $(x_{0},x_{1},\cdots,x_{s-1})^{T}$ and such operation can be achieved by a matrix $R\in\mathbb{F}_{q}^{s\times s}$ \textit{i.e.},
    \begin{displaymath}
        \mathrm{rev}\left((x_{0},x_{1},\cdots,x_{s-1})^{T}\right)=R(x_{0},x_{1},\cdots,x_{s-1})^{T}=(x_{s-1},x_{s-2},\cdots,x_{0})^{T}
    \end{displaymath}
    for any $(x_{0},x_{1},\cdots,x_{s-1})\in\mathbb{F}_{q}^{s} $. Note that $R^{2}=I_{s}$.

    Denote by $D_{j}=\mathrm{diag}(v_{1,j},v_{2,j},\cdots,v_{s,j})$ and $D_{j}^{\perp}=\mathrm{diag}(v^{-1}_{s,j},v^{-1}_{s-1,j},\cdots,v_{1,j}^{-1})$. Then we have the following duality result.
    
	\begin{theorem}\label{c1}
    Let $\eta$ be a nonzero differential satisfying $\mathrm{v}_{P_j}(\eta)=-s$ for $1\leq j\leq r$, and set
    \begin{displaymath}
        G^{\vee}=(\eta)-G+E.
    \end{displaymath}
    Consider local expansions
    \begin{displaymath}
        \eta=t_j^{-s}u_j(t_j)\,dt_j,\qquad
        u_j(t_j)=\sum_{c\geq0}u_{j,c}t_j^c,\quad u_{j,0}\neq0,
    \end{displaymath}
    and define the upper-triangular Toeplitz matrix
    \begin{displaymath}
        U_j=\begin{pmatrix}
        u_{j,0}&u_{j,1}&\cdots&u_{j,s-1}\\
        0&u_{j,0}&\cdots&u_{j,s-2}\\
        \vdots&\ddots&\ddots&\vdots\\
        0&\cdots&0&u_{j,0}
        \end{pmatrix}.
    \end{displaymath}
    For $h\in\mathcal{L}(G^\vee)$, let
    $\mathrm{rev}( \bm h_j)=(h_{j,s-1},\ldots,h_{j,0})^T$ denote the reverse vector of its first $s$ local coefficients at $P_j$. Then the generator matix of $  \left(\mathcal{C}^{(s)}_{\mathcal{L}}(D,G,V)\right)^\perp   $ can be given by
    \begin{displaymath}
        \left\{\left(D_j^{-1}U_j\mathrm{rev}(\bm h_j)\right)_{j=1}^{r}:h\in\mathcal{L}(G^\vee)\right\}.
    \end{displaymath}
    \end{theorem}
	\begin{proof}

Based on the decomposition $\omega=h\eta$ in Lemma \ref{lp}, we have the following local expansions at $P_{j}$:
\begin{displaymath}
    \eta = t^{-s}_{j}u_{j}(t_{j})dt_j,\ u_j(t_j) = \sum_{c\geq 0} u_{j,c}t^{c}_{j},u_{j,0}\neq 0,
\end{displaymath}
and $h =\sum_{b\geq 0} h_{j,b}t^{b}_{j}$.
    
	Each coefficient can be represented by $$c_{j,-a-1}(\omega)= [t_{j}^{s-1-a}]  (u_{j}(t_{j})h(t_{j})),$$ where $ [t_{j}^{s-1-a}]  (u_{i}(t_{j})h(t_{j}))$ denotes the coefficient of $u_{i}(t_{j})h(t_{j})$ at monomial $t_{j}^{s-1-a}$. By direct calculation, the coefficients can be represented explicitly by
    \begin{displaymath}
    \begin{split}
        &[t_{j}^{s-1}](u_{i}(t_{j})h(t_{j}))=u_{j,0}h_{j,s-1}+\cdots+u_{j,s-1}h_{j,0}  \\
        &\cdots\\
         &[t_{j}](u_{i}(t_{j})h(t_{j}))=u_{j,0}h_{j,1}+u_{j,1}h_{j,0}\\
        &[t_{j}^{0}](u_{i}(t_{j})h(t_{j}))=u_{j,0}h_{j,0}.\\
        \end{split}
    \end{displaymath}

 Note that codewords naturally reverse by the duality structure. Consider the reverse vector $\mathrm{rev}(\bm h_{j})=R(\bm h_{j})=(h_{j,s-1},h_{j,s-2},\cdots,h_{j,0})^{T}$ for each $1\leq j\leq r$. Then each column of generator matrix of differential HAG code can be decomposed into
	\begin{displaymath}
		D_{j}^{-1}c_{j}(\omega)=D_{j}^{-1}U_{j}R\bm h_{j},
	\end{displaymath}
	where each $U_{j}$ is given by
	\begin{equation}\label{m1}
		U_{j}=\begin{pmatrix}
			u_{j,0}&u_{j,1}&\cdots&u_{j,s-1}\\
			0&u_{j,0}&\cdots&u_{j,s-2}\\
			\vdots&\ddots&\ddots&\vdots\\
			0&\cdots&0&u_{j,0}
		\end{pmatrix}
	\end{equation}
for $1\leq j\leq r$.
        \end{proof}

	\begin{example}\label{e1}
		Take a rational function field $\mathbb{F}_{q}(x)$. Choose places $P_{j}: x=\alpha_{j},1\leq j\leq r$ and divisor $G=(k-1)P_{\infty}$. Let $A(x)=\prod^{r}_{j=1}(x-\alpha_{j})^{s}$ and $\eta=\frac{dx}{A(x)}$. Then we have
		\begin{displaymath}
			(\eta)=-sD+(sr-2)P_{\infty}, \ G^\vee=(sr-k-1)P_{\infty}.
		\end{displaymath}
		Consider local expressions. For each rational place $P_{j}$,
		\begin{displaymath}
			\eta=(x-\alpha_{j})^{-s}A_{j}(x)^{-1}d(x-\alpha_{j}),
		\end{displaymath}
		Then we have 
		\begin{displaymath}
			c_{a+1,j}(g)=v_{a+1,j}^{-1}\partial^{(s-1-a)}\left(\frac{g}{A_{j}}\right)(\alpha_{j})
		\end{displaymath}
		for $g\in\mathcal{L}((sr-k-1)P_{\infty})$. This is the result given in \cite{99}.
	\end{example}

     \begin{example}[Elliptic Function Field]\label{e2}
    Let $F=\mathbb{F}_{5}(\mathcal{X})$ be an elliptic function field with the defining equation:
    \begin{displaymath}
        \mathcal{X}/\mathbb{F}_{5}:\quad y^{2}=x^{3}+x+1
    \end{displaymath}
     with a unique point at infinity $O$. Consider three rational places
    \begin{displaymath}
        P_{1}=(0,1),\qquad P_{2}=(2,1),\qquad P_{3}=(3,1),
    \end{displaymath}
    divisor $D=\sum^{3}_{j=1}P_{j}$, $E=2D$, $G=3O$, and $V=(1)_{2\times 3}$.
    Since $(x)_{\infty}=2O$ and $(y)_{\infty}=3O$, the Riemann--Roch space is
    \begin{displaymath}
        \mathcal{L}(3O)=\langle 1,x,y\rangle_{\mathbb{F}_{5}}.
    \end{displaymath}

    For $P_{j}=(\alpha_{j},1)$, we have $2y(P_j)=2\neq 0$ and $t_{j}=x-\alpha_{j}$ is a local parameter. From
    \begin{displaymath}
        2y\frac{dy}{dx}=3x^{2}+1,
    \end{displaymath}
we have the local expansions for the function $y$ at each place $P_{j}, 1\leq j\leq 3$.
     \begin{table}[htbp]  
     \caption{ local expansions for $y$ }
  \centering
    \begin{tabular}{|c|c|c|}
    \toprule
        &$\alpha_{j}$&$y$\\ \midrule
       $P_{1}$&$0$&$1+3t_{1}+3t_{1}^{2}+O(t_{1}^{3})$\\\midrule
        $P_{2}$&$2$&$1+4t_{2}+O(t_{2}^{3})$\\\midrule
        $P_{3}$&$3$&$1+4t_{3}+4t_{3}^{2}+O(t_{3}^{3}).$\\
        \bottomrule
    \end{tabular}
    \end{table}\label{t1}
    
    Then the HAG code
    $\mathcal{C}^{(2)}_{\mathcal{L}}(D,3O,V)$ with respect to the basis
    $\mathcal{L}(3O)$ can be generated by
    \begin{displaymath}
        G_{1}=
        \begin{pmatrix}
        1&1&1\\
        0&0&0\\
        \end{pmatrix},
          G_{x}=
        \begin{pmatrix}
        0&2&3\\
        1&1&1\\
        \end{pmatrix},
          G_{y}=
        \begin{pmatrix}
        1&1&1\\
        3&4&4\\
        \end{pmatrix},
    \end{displaymath}
    \textit{i. e.}, $\mathcal{C}^{(2)}(D,3O,V)=\mathrm{span}_{\mathbb{F}_{5}}\{G_{1},G_{x},G_{y}\}$.
    
     In particular, the general lower bound gives
    $d_{NRT}\geq 6-\deg(G)=3$. On the other hand, for $f=1-y$ we obtain the codeword
    \begin{displaymath}
        \mathrm{ev}^{(2)}_{D,G,V}(1-y)
        =\big((0,2)^{T},(0,1)^{T},(0,1)^{T}\big),
    \end{displaymath}
    whose NRT weight is $3$. Thus we have
    \begin{displaymath}
        d_{NRT}\left(\mathcal{C}^{(2)}_{\mathcal{L}}(D,3O,V)\right)=3.
    \end{displaymath}

    This example also makes the generalized NRT weights explicit. For an elliptic curve,
    $\gamma_{1}=0$ and $\gamma_{m}=m$ for $m\geq 2$. Hence, the generalized-weight theorem yields
    \begin{displaymath}
        d_{1}^{NRT}\left(\mathcal{C}^{(2)}_{\mathcal{L}}(D,3O,V)\right)\geq 3,\quad d_{2}^{NRT}\left(\mathcal{C}^{(2)}_{\mathcal{L}}(D,3O,V)\right)\geq 5,\quad d_{3}^{NRT}\left(\mathcal{C}^{(2)}_{\mathcal{L}}(D,3O,V)\right)\geq 6.
    \end{displaymath}
    The first equality has already been shown. For
    $U=\langle x,y-1\rangle_{\mathbb{F}_{5}}$, we have
    $e_{1}(U)=1$ and $e_{2}(U)=e_{3}(U)=0$, so 
$w_{2}^{NRT}(U)=6-1=5$. Finally,  $d_{3}^{NRT}\left(\mathcal{C}^{(2)}_{\mathcal{L}}(D,3O,V)\right)=6$ holds trivially. Consequently, we have
    \begin{displaymath}
        \left(d_{1}^{NRT}\left(\mathcal{C}^{(2)}_{\mathcal{L}}(D,3O,V)\right),d_{2}^{NRT}\left(\mathcal{C}^{(2)}_{\mathcal{L}}(D,3O,V)\right),d_{3}^{NRT}\left(\mathcal{C}^{(2)}_{\mathcal{L}}(D,3O,V)\right)\right)=(3,5,6).
    \end{displaymath}

 Next, we explicitly compute the dual code. Consider the invariant differential
    \begin{displaymath}
        \omega_{0}=\frac{dx}{2y}
    \end{displaymath}
   and it satisfies $(\omega_{0})=0$. Moreover,
    \begin{displaymath}
        (y-1)=P_{1}+P_{2}+P_{3}-3O=D-3O.
    \end{displaymath}
    Define
    \begin{displaymath}
        \eta=\frac{\omega_{0}}{(y-1)^{2}}=\frac{dx}{2y(y-1)^{2}}.
    \end{displaymath}
    Then
    \begin{displaymath}
        (\eta)=-2D+6O,
        \qquad
        G^{\vee}=(\eta)-G+E=3O=G.
    \end{displaymath}

    Consider the expansion
    \begin{displaymath}
\eta=t_{j}^{-2}\left(u_{j,0}+u_{j,1}t_{j}+O(t_{j}^{2})\right)dt_{j}.
    \end{displaymath}
    If $y=1+a_jt_j+b_jt_j^2+O(t_j^3)$, then
    \begin{displaymath}
        u_{j,0}=(2a_j^2)^{-1},\qquad
        u_{j,1}=-u_{j,0}\left(a_j+2b_ja_j^{-1}\right).
    \end{displaymath}
    Using the above expansions of $y$, we have
    \begin{displaymath}
        (u_{1,0},u_{1,1})=(2,0),\qquad
        (u_{2,0},u_{2,1})=(3,3),\qquad
        (u_{3,0},u_{3,1})=(3,2).
    \end{displaymath}
    Hence the equivalence transformation matrices in Theorem \ref{c1} can be given by
    \begin{displaymath}
        U_{1}=\begin{pmatrix}2&0\\0&2\end{pmatrix},\qquad
        U_{2}=\begin{pmatrix}3&3\\0&3\end{pmatrix},\qquad
        U_{3}=\begin{pmatrix}3&2\\0&3\end{pmatrix}.
    \end{displaymath}
    Since $V$ is the all-one multiplier matrix, we have $T_{j}=U_{j}$ for $D_{j}=I_{2}$.
    Then the dual code $ \mathcal{C}^{(2)}_{\mathcal{L}}(D,3O,V)^{\perp}  $ is generated by
    \begin{displaymath}
        H_{1}=
        \begin{pmatrix}
        0&3&2\\
        2&3&3\\
        \end{pmatrix},
         H_{x}=
        \begin{pmatrix}
        2&4&4\\
        0&1&4\\
        \end{pmatrix},
         H_{y}=
        \begin{pmatrix}
        1&0&4\\
        2&3&3\\
        \end{pmatrix},
    \end{displaymath}
\textit{i. e.}, $\mathcal{C}^{(2)}(D,3O,V)^{\perp}=\mathrm{span}_{\mathbb{F}_{5}}\{H_{1},H_{x},H_{y}\}$.
    One checks directly that each of $H_{1},H_{x},H_{y}$ is orthogonal to
    $G_{1},G_{x},G_{y}$ under the Euclidean inner product.
    
    This illustrates that although
    $G^{\vee}=G$, the Euclidean dual is obtained from the reversed evaluation code through
    nontrivial local triangular blocks $U_{j}, 1\leq j\leq r$.
    \end{example}

	\begin{remark}
		In the construction of HRS codes, the evaluation function is chosen from the Riemann-Roch space $\mathcal{L}((k-1)P_{\infty})$. For an arbitrary function field,  the dual divisor $G^\vee$
		is generally not in the form $(sr-k-1)P_{\infty}$. Therefore, the dual of an HRS code is not necessarily a HRS code, while it needs further modification of local matrices $U_{j}, 1\leq j\leq r$. 
	\end{remark}
	
	\subsection{ A Normalized Choice of Differentials  }
	In the previous subsection, we have shown that the dual of $  \mathcal{C}^{(s)}_{\mathcal{L}}(D,G,V)$ is reversely pairwise equivalent to 
	$ \mathcal{C}^{(s)}_{\mathcal{L}}(D,G^\vee,V^\perp)$ by upper-triangular matrices $U_{j}, 1\leq j\leq r$. In Example \ref{e1}, the given differential $\eta=dx/A(x)$ induces the modification of local matrices $U_{j}, 1\leq j\leq r$.

    In this subsection, we shall provide a modified construction such that each $U_{j}\equiv I_{s}$, \textit{i. e.}, $u_{j}(t_{j})\equiv 1$ (mod $t_{j}^{s}$) for $1\leq j\leq r$. Through this modification, we can obtain a simplified version of reverse duality theorem.

Denote by $E=sD=s\sum^{r}_{j=1}P_{j}$ as before and $PP(E)$ the principal differential space which consists of differentials with indices of poles less than or equal to $s$, \textit{i. e.},
\begin{displaymath}
    PP(E)=\bigoplus^{r}_{j=1}\bigoplus^{-1}_{k=-s}\mathbb{F}_{q}t_{j}^{k}dt_{j}.
\end{displaymath}
Note that $\dim_{\mathbb{F}_{q}}PP(E)=sr$.

Define the natural map $\pi:\Omega_{F}(-E)\to PP(E)$. Then we have the following Lemma.
\begin{lemma}\label{l1}
For $\pi_j=\sum_{k=-s}^{-1}a_{j,k}t_j^kdt_j$, define
\begin{displaymath}
    S=\left\{(\pi_j)_{j=1}^r\in PP(E):
    \sum_{j=1}^r\operatorname{res}_{P_j}(\pi_j)=0\right\}.
\end{displaymath}
Then $\operatorname{Im}(\pi)=S$.
\begin{proof}
The global residue theorem gives $\operatorname{Im}(\pi)\subseteq S$. Let $K$ be a canonical divisor. Since
$\Omega_F(-E)\simeq\mathcal L(K+E)$,
\begin{displaymath}
    \dim\Omega_F(-E)=\mathfrak g_F-1+\deg E.
\end{displaymath}
The kernel of $\pi$ is the space $\Omega_F(0)$ of holomorphic differentials, which has dimension $\mathfrak g_F$. Hence
\begin{displaymath}
    \dim\operatorname{Im}(\pi)=\deg E-1=sr-1.
\end{displaymath}
On the other hand, $PP(E)$ has dimension $sr$, and the residue map
$PP(E)\to\mathbb F_q$ is surjective. Therefore $S$ also has dimension $sr-1$, and the inclusion is an equality.
\end{proof}
\end{lemma}

	\begin{lemma}\label{l2}
		For any local parameters $t_{1},\cdots,t_{r}$, there exists a $\mathbb{F}_{q}$-rational differential $\eta$ regular at $\mathbb{P}_{F}\setminus\mathrm{Supp}(D)$ and 
		\begin{displaymath}
			\eta=(t_{j}^{-s}+O(1))dt_{j}=t_{j}^{-s}(1+O(t^{s}_{j}))dt_{j}
		\end{displaymath}
		for any $1\leq j\leq r$ and $s\geq 2$.
		\begin{proof}
			For each place $P_{j}: 1\leq j\leq r$, we fix a principal part $t_{j}^{-s}dt_{j}$. For $s\geq 2$, the local residues are identically zero. By Lemma \ref{l1} and strong approximation theorem in \cite{1}, such a global differential exists.
        \end{proof}
	\end{lemma}
	
	\begin{theorem}\label{cc}
		Assume $s\geq 2$ and fix a differential $\eta$ in Lemma \ref{c1}. The dual of HAG code $  \mathcal{C}^{(s)}_{\mathcal{L}}(D,G,V)$ is exactly the reverse code $\mathrm{rev}\left( \mathcal{C}^{(s)}_{\mathcal{L}}(D,G^\vee,V^\perp)\right)$.
		where $V^{\perp}=(v^{-1}_{s-a-1,j})_{0\leq a\leq s-1,1\leq j\leq r}$.
		\begin{proof}
			It follows from $u_{j}(t_{j})\equiv 1$ (mod $t_{j}^{s}$) that $U_{j}\equiv I_{s}$ for any $1\leq j\leq r$. Then 
           each codeword $D_{j}^{-1}c_{j}(\omega)$ and reverse vector $\mathrm{rev}(\bm h_{j})$ satisfy
            \begin{displaymath}
                D_{j}^{-1}c_{j}(\omega)=D_{j}^{-1} U_{j}\mathrm{rev}( \bm h_{j})=D_{j}^{-1} \mathrm{rev}(\bm h_{j})
            \end{displaymath}
            for any $1\leq j\leq r$. Then we have the desired result.
		\end{proof}
	\end{theorem}

	\subsection{Self-duality via Reverse Bilinear Form}
	As we discussed above, the duals of HAG codes naturally reverse via residue theorem. To discuss the self-orthogonality and self-duality, we need the following reverse bilinear form.
	
	\begin{definition}\label{dd}
		For any $X,Y\in\mathbb{F}_{q}^{s\times r }$, reverse bilinear form is given by
		\begin{displaymath}
			[X,Y]_{\mathrm{rev}}=\sum^{r}_{j=1}\sum^{s-1}_{a=0}X_{a+1,j}Y_{s-a,j}.
		\end{displaymath}
        \end{definition}
    For any code $\mathcal{C}\in\mathbb{F}_{q}^{s\times r}$, we define 
    \begin{displaymath}
        C^{\perp_{\mathrm{rev}}} :=\{Y :[X,Y]_{\mathrm{rev}} =0\ \text{for all}\ X\in\mathcal{C}\}=\mathrm{rev}(\mathcal{C}^{\perp})
    \end{displaymath}

        \begin{lemma}\label{ll}
        Assume $v_{a+1,j}v_{s-a,j}=\lambda_{j}, 0\leq a\leq s-1$. For any $f,g\in\mathcal{L}(G)$, we have
		\begin{displaymath}
			[\mathrm{ev}^{(s)}_{D,G,V}(f), \mathrm{ev}^{(s)}_{D,G,V}(g)  ]=\sum^{r}_{j=1}\lambda_{j}[t_{j}^{s-1}](fg).
		\end{displaymath}
        \end{lemma}
        \begin{proof}
              By definition, we have 
        \begin{displaymath}
        \begin{split}
            	[\mathrm{ev}^{(s)}_{D,G,V}(f),\mathrm{ev}^{(s)}_{D,G,V}(g)  ]&=\sum^{r}_{j=1}\sum^{s-1}_{a=0} v_{a+1,j}v_{s-a,j}f_{j,a}g_{s-a-1,j} \\
                &=\sum^{r}_{j=1}\sum^{s-1}_{a=0} \lambda_{j}f_{j,a}g_{s-a-1,j}\\
                &=\sum^{r}_{j=1}\lambda_{j}[t_{j}^{s-1}](fg).\\
                \end{split}
                \end{displaymath}
        \end{proof}
		
		\begin{proposition}
			If	$\dim  \mathcal{C}^{(s)}_{\mathcal{L}}(D,G,V)=\frac{sr}{2}$, the HAG code $\mathcal{C}^{(s)}_{\mathcal{L}}(D,G,V)$ is reverse self-dual if and only if 
			\begin{displaymath}
				\sum^{r}_{j=1}\lambda_{j}[t_{j}^{s-1}](fg)=0\ \text{for all}\ f,g\in\mathcal{L}(G).
			\end{displaymath}
            \begin{proof}
            Set $k=\ell(G)-\ell(G-E)=\dim_{\mathbb{F}_{q}}\mathcal{C}_{\mathcal{L}}^{(s)}(D,G,V)$. If  $\dim  \mathcal{C}^{(s)}_{\mathcal{L}}(D,G,V)=\frac{sr}{2}$ is reversely self-dual, we have $$\dim_{\mathbb{F}_{q}} \mathcal{C}_{\mathcal{L}}^{(s)}(D,G,V)^{\perp_{\mathrm{rev}}}=rs-k=k.$$ Combining with Lemma \ref{ll} and $k=\frac{sr}{2}$, we have the desired results.
            \end{proof}
		\end{proposition}

        Then the result of self-duality of HRS codes in \cite{99} can be immediately deduced in our construction framework.
		\begin{example}(Rational Function Field)
			Consider $\mathbb{F}_{q}(x)$ and take two monomials $f(x)=x^{a}, g(x)=x^{b},$ for $ a,b\leq k-1$. Then HRS code $\mathcal{C}^{(s)}_{\mathcal{L}}(D,(k-1)P_{\infty},V)$ is reverse self-dual if and only if $2k=sr$ and
			\begin{displaymath}
				\binom{a+b}{s-1}\sum^{r}_{j=1}\lambda_{j}\alpha_{j}^{a+b-s+1}=0.
			\end{displaymath}
           if $a+b\geq s-1$. In particular, if $a+b<s-1$, the summation is defined as $0$.
		\end{example}

		\subsection{Self-duality via Residue Bilinear Form}
			We now define the second bilinear form for a fixed differential $\eta$ satisfying $\mathrm{v}_{P_j}(\eta)=-s$ for all $j$. Let
        \begin{displaymath}
            A_j=\mathbb{F}_q[t_j]/(t_j^s),\qquad
            \phi_j\!\left(\sum_{i=0}^{s-1}a_it_j^i\right)=(a_0,\ldots,a_{s-1})^T.
        \end{displaymath}
        Define the local residue pairing by
        \begin{displaymath}
            B_{\eta,j}(\phi_j(a),\phi_j(b))=\mathrm{res}_{P_j}(ab\eta).
        \end{displaymath}
 Furthermore, 
       the bilinear form can be given by $$ B_{\eta,j}
^{(V)}(\phi_{j}(a),\phi_{j}(b))=B_{\eta,j}(D_j^{-1}\phi_{j}(a)   ,D_j^{-1}\phi_{j}(b)  )$$ with respect to the multiplier matrix $V$. Then we have the global differential bilinear form:
		\begin{displaymath}
			B ^{(V)} _{\eta}=\sum^{r}_{j=1}B^{(V)}_{\eta,j}.
		\end{displaymath}
For any code $\mathcal{C}\in\mathbb{F}_{q}^{s\times r}$, we define
        $$\mathcal{C}^{\perp_{B^{(V)}_{\eta}}} =\{Y:B^{(V)}_{\eta}(X,Y)= 0, \text{for all}\ X\in\mathcal{C}\}.$$
    
		For the given basis $1,t_j,\ldots,t_j^{s-1}$, the local Gram matrix can be given by
        \begin{displaymath}
            H_j=(u_{j,s-1-a-b})_{0\leq a,b\leq s-1},
        \end{displaymath}
        where $u_{j,c}=0$ for $c<0$. Hence the determinant can be calculated by
        \begin{displaymath}
            \det H_j=(-1)^{s(s-1)/2}u_{j,0}^s\neq0.
        \end{displaymath}
        Thus both the local forms and $B_\eta^{(V)}$ are nondegenerate.

        The following theorem gives self-orthogonality and self-duality with respect to the multiplier-adjusted residue form.
		\begin{theorem}[Self-orthogonality and self-duality for the residue form]\label{c3}
        
 			Suppose $\eta\neq0$ and $v_{P_j}(\eta)=-s$ for every $P_j\in\operatorname{Supp}(D)$. Then  
			\begin{enumerate}
				\item $\mathcal{C}^{(s)}_{\mathcal{L}}(D,G,V)$ is self-orthogonal with respect to $B_{\eta}^{(V)}$ if $(\eta)\geq 2G-E$.
				\item  $\mathcal{C}^{(s)}_{\mathcal{L}}(D,G,V)$ is self-dual with respect to $B_{\eta}^{(V)}$ if $(\eta)=2G-E$.
			\end{enumerate}
			\begin{proof}
            Denote by $X_{j}(f)$ the $j$-th column of $X(f)=\mathrm{Im}(\mathrm{ev}_{D,G,V}^{(s)}(f))$.
				For any $f,g\in\mathcal{L}(G)$, we have $(fg\eta)\geq (\eta)-2G\geq -E.$ Thus $fg\eta$ has no poles outside $\operatorname{Supp}(D)$, and 
				\begin{displaymath}
					\sum^{r}_{j=1}B_{\eta,j}(D_{j}^{-1}X_{j}(f),  D_{j}^{-1}X_{j}(g) )=\sum^{r}_{j=1}\mathrm{res}_{P_{j}}(fg\eta)=0
				\end{displaymath}
               by the global residue theorem.
                
				If $2G=(\eta)+E$, then $K=(\eta)$ satisfies $K-G=G-E$.
                By the Riemann-Roch theorem, we have
                \begin{displaymath}
                \begin{split}
                    \dim_{\mathbb{F}_{q}}\mathcal{C}^{(s)}_{\mathcal{L}}(D,G,V)&=\ell(G)-\ell(G-E)\\
                    &=\ell(G)-\ell(K-G)\\
                    &=\deg G+1-\mathfrak{g}_{F}.
                    \end{split}
                \end{displaymath}
                Since $2\deg G-sr=2\mathfrak g_F-2$, $\deg G+1-\mathfrak{g}_{F}=\frac{sr}{2}$. Because $B_{\eta}^{(V)}$ is nondegenerate, self-orthogonality together with this half-dimension condition implies self-duality.
			\end{proof}
		\end{theorem}
		\begin{corollary}
			Suppose $s\geq 2$ and the existence of the differential $\eta=(t_{j}^{-s}+O(1))dt_{j}, 1\leq j\leq r$ satisfying $2G=(\eta)+E.$ If we further assume $v_{i,j}v_{s-i+1,j}=1,$ then $\mathcal{C}^{(s)}_{\mathcal{L}}(D,G,V)$ is reverse self-dual with regard to the reverse bilinear form $[\cdot,\cdot]_{\mathrm{rev}}$.
		\end{corollary}
	
	\begin{remark}
		Denote by $K$ the canonical class of $F/\mathbb{F}_{q}$. By the equality $2G=(\eta)+E$, we have $2[G]=[K+E]$ in the Picard group $\mathrm{Pic}(F)$, which also means that the degree $\deg(K+E)=2\mathfrak{g}_{F}-2+sr$ is even, \textit{i. e.}, $sr$ is even. Besides, we also require $[K+E]$ is divisible by $2$ in $\mathrm{Pic}(F)$.
	\end{remark}

	\section{Function Field towers and asymptotic NRT parameters}
	\label{sec:4}
	Fix a positive integer $s\in\mathbb{N}$. Let $\mathcal{F}=(F_{i}/\mathbb{F}_{q})_{i\geq1}$ be a function field tower with genus $\mathfrak{g}_{F_{i}}\to \infty$ and 
	\begin{displaymath}
		\lim_{i\to\infty}\mathfrak{g}_{i}=\infty\ \text{and}\ \lim_{i\to \infty}\frac{N_{i}}{\mathfrak{g}_{F_{i}}}=\lambda(\mathcal{F}).
	\end{displaymath}
    Note that $\lambda(\mathcal{F})\leq A(q)$. In particular, we have $\lambda(F)=A(q)$ if $\mathcal{F}$ is optimal.
     In the following construction,
	we choose $Q_{i}\in\mathbb{P}_{F_{i}}$ as a unique rational pole and label $r_{i}=N_{i}-1$ rational places as $D_{i}$. Let $G_{i}=m_{i}Q_{i}$, $E_{i}=sD_{i}$ and $n_{i}=sr_{i}$.
	\begin{theorem}\label{c4}
		Let $\mathcal{F}=(F_{i}/\mathbb{F}_{q})_{i\geq 1}$ be a function field tower with
        \begin{displaymath}
            \mathfrak{g}_{i}\to\infty,\ \frac{N_{i}}{\mathfrak{g}_{i}}\to \lambda(\mathcal{F}).
        \end{displaymath}
        Then we can construct a family of $q-$ary HAG codes whose  information rate $R$ and relative distance $\delta$ satisfy
		\begin{displaymath}
			R\geq 1-\frac{1}{s\lambda(\mathcal{F})}-\delta.
		\end{displaymath}
		\begin{proof}
From Theorem \ref{tt}, we can construct a family of $q$-ary HAG codes with parameters $[n_{i},k_{i},d_{i}]_{q}$ satisfying
\begin{displaymath}
   k_{i}=\ell(G_{i})\geq m_{i}+1-\mathfrak{g}_{i}\ \text{and}\ d_{i}\geq n_{i}-k_{i}+1-\mathfrak{g}_{i} .
\end{displaymath}

 Then it follows that
\begin{displaymath}
\frac{d_{i}}{n_{i}}\geq 1-\frac{k_{i}}{n_{i}}+\frac{1-\mathfrak{g}_{i}}{n_{i}}.
\end{displaymath}
Without loss of generality, we assume the existence of the two limits:
\begin{displaymath}
    R=\lim_{i\to\infty}\frac{k_{i}}{n_{i}}\ \text{and}\ \delta=\lim_{i\to\infty}\frac{d_{i}}{n_{i}}.
\end{displaymath}
Then we have 
\begin{displaymath}
\begin{split}
\delta+R  &\geq 1+\lim_{i\to\infty}\frac{1-\mathfrak{g}_{i}}{n_{i}}\\
    &=1-\frac{1}{s\lambda(\mathcal{F})}.
    \end{split}
\end{displaymath}
		\end{proof}
	\end{theorem}
	
	\begin{corollary}
		\begin{enumerate}
			\item 	If $q$ is a square, then we can construct a family of $q-$ary HAG codes whose information rate $R$ and relative distance $\delta$ satisfy
			\begin{displaymath}
				R\geq 1-\frac{1}{s(\sqrt{q}-1)}-\delta.
			\end{displaymath}
			\item If $q=p^{2m+1}$ and $\lambda(\mathcal{F})=A(q)$, then we can construct a family of $q-$ary HAG codes whose information rate $R$ and relative distance $\delta$ satisfy
			\begin{displaymath}
				R\geq 1-\frac{1}{2s}\left(\frac{1}{p^{m}-1}+\frac{1}{p^{m+1}-1}\right)-\delta.
			\end{displaymath}
			\item If $\ell$ is a prime, $\lambda(\mathcal{F})=A(q)$, $\ell\mid (p-1)$ and $p>4\ell+1$, then we can construct a family of $q=p^{\ell}-$ary HAG codes whose information rate $R$ and relative distance $\delta$ satisfy
			\begin{displaymath}
				R\geq 1-\frac{\ell-1}{s\left(\sqrt{\ell(p-1)}-2\ell\right)}-\delta.
			\end{displaymath}
		\end{enumerate}
		\begin{proof}
			If $q$ is a square, then we can choose the Garcia-Stichtenoth tower such that $\lambda(\mathcal{F})=A(q)=\sqrt{q}-1$. Then we have 
            \begin{displaymath}
                R\geq 1-\frac{1}{sA(q)}-\delta=1-\frac{1}{s(\sqrt{q}-1)}-\delta.
            \end{displaymath}
            If $q=p^{2m+1}$, then we have 
			\begin{displaymath}
				A(p^{2m+1})\geq \frac{2(p^{m+1}-1)}{p+1+ \frac{p-1}{p^{m}-1}  }=\frac{2}{\frac{1}{p^{m}-1}+\frac{1}{p^{m+1}-1}}.
			\end{displaymath}
			and
              \begin{displaymath}
                R\geq 1-\frac{1}{sA(p^{2m+1})}-\delta\geq 1-\frac{1}{2s}\left(\frac{1}{p^{m}-1}+\frac{1}{p^{m+1}-1}\right)-\delta.
            \end{displaymath}
            If $q=p^\ell$, $\ell\mid(p-1)$, and $p>4\ell+1$, then we have $$  A(p^{\ell})\geq\frac{\sqrt{\ell(p-1)}-2\ell}{\ell-1}$$ and 
             \begin{displaymath}
                R\geq 1-\frac{1}{sA(p^{\ell})}-\delta\geq 1-\frac{\ell-1}{s\left(\sqrt{\ell(p-1)}-2\ell\right)}-\delta.
            \end{displaymath}
		\end{proof}
	\end{corollary}
	\section{Conclusion and possible future work}
	\label{sec:5}
	In this paper, we provide two constructions of HAG codes, \textit{i. e.}, evaluation and differential HAG code via local expansions. Our constructions extend HRS codes to the general function fields.  We also completely determine the duality via residue theorem. In particular, the equivalence relation and reverse duality in \cite{99} can also be illustrated explicitly via local expansions. Within the framework of function fields, we can also simplify the prior construction of HAG codes by carefully choosing a special type of differentials. We also determine self-orthogonal and self-dual conditions via two types of bilinear forms. Finally, we also obtained some asymptotic results via some function field towers and Ihara's constant. As for possible future work, it is interesting to explore further novel constructions of AG codes.

\end{document}